\documentclass[showpacs,preprintnumbers,amsmath,amssymb,11pt]{article}
\usepackage{graphicx,amsfonts,amssymb,amsthm,amsmath,graphics,epsfig,pstricks,fancyhdr,fancybox}
\usepackage{dcolumn}
\usepackage{bm}
\usepackage{rotating}
 \newtheorem{stdef}{Definition}
 \newtheorem{prop}{Proposition}

\begin{document}

\title{\LARGE \textbf{Pricing and Hedging Asian Basket Options\\
with Quasi-Monte Carlo Simulations}}

\author{Nicola \textsc{Cufaro Petroni}\\
Dipartimento di Matematica and \textsl{TIRES}, Universit\`a di Bari\\
\textsl{INFN} Sezione di Bari\\
via E. Orabona 4, 70125 Bari, Italy\\
\textit{email: cufaro@ba.infn.it}\\
\\
Piergiacomo \textsc{Sabino}\\
Dipartimento di Matematica, Universit\`a di Bari\\
via E. Orabona 4, 70125 Bari, Italy \\
\textit{email: sabino@dm.uniba.it}}

\maketitle

\abstract{
% The abstract\index{abstract} should summarize the contents of the  paper
% in at least 70 and at most 150 words;
%
In this article we consider the problem of pricing and hedging
high-dimensional Asian basket options by Quasi-Monte Carlo
simulation. We assume a Black-Scholes market with time-dependent
volatilities and show how to compute the deltas by the aid of the
Malliavin Calculus, extending the procedure employed by Montero and
Kohatsu-Higa~\cite{MKH2003}. Efficient path-generation algorithms,
such as Linear Transformation and Principal Component Analysis,
exhibits a high computational cost in a market with time-dependent
volatilities. We present a new and fast Cholesky algorithm for block
matrices that makes the Linear Transformation even more convenient.
Moreover, we propose a new-path generation technique based on a
Kronecker Product Approximation. This construction returns the same
accuracy of the Linear Transformation used for the computation of
the deltas and the prices in the case of correlated asset returns
while requiring a lower computational time. All these techniques can
be easily employed for stochastic volatility models based on the
mixture of multi-dimensional dynamics introduced by Brigo et
al.~\cite{BMR,BMR04}.}
%%%%%%%%%%%%%%%%%%%%%%%%%%%%%%%%%%%%%%%%%%%%%%%%%%%%%%%%%%%%%%%%%%%%%%
%
%%%%%%%%%%%%%%%%%%%%%%%%%%%%%%%%%%%%%%
\section {Introduction and Motivation}\label{sec:intro}
%This paper aims to extend the studies presented in a series of
%manuscripts by Dahl and Benth~\cite{DB2002} and Imai and Tan
%\cite{IT2006}.

In a recent paper Dahl and Benth~\cite{DB2002} have investigated the
efficiency and the computational cost of the Principal Component
Analysis (PCA) used in the Quasi-Monte Carlo (QMC) simulations for
the pricing of high-dimensional Asian basket options in a
multi-dimensional Black-Scholes (BS) model with constant
volatilities. In particular they have shown the essential role of
the Kronecker product for a fast implementation as well as for the
analysis of variance (ANOVA) in order to identify the effective
dimension (see below). Indeed the convergence rate of the QMC method
is $O\left(N^{-1}\log^dN\right)$, where $N$ is the number of
simulation trials, and $d$ the nominal dimension of the problem.
This implies that the theoretically higher asymptotic convergence
rate of QMC could not be achieved for practical purposes in high
dimensions. On the other hand, some  applications in finance (see
Paskov and Traub~\cite{PT1995}) have shown that QMC provides a
higher accuracy than standard Monte Carlo (MC), even for high
dimensions.

To explain the success of QMC in high dimensions Caflisch et al.
\cite{CMO1997} have introduced two notions of effective dimensions
based on the ANOVA of the integrand function. Consider an integrand
function $f$ and a MC problem with nominal dimension $d$. Let
$\mathcal{A}=\{1,\dots,d\}$ denote the labels of the input variables
of the function $f$: then the effective dimension of $f$, in the
\emph{superposition sense}, is the smallest integer $d_S$ such that
$\sum_{|u|\le d_S} \sigma^2(f_u)\ge p\sigma^2(f)$, where $f_u$ is a
function with variables in the set $u\subseteq \mathcal{A}$,
$\sigma^2(\cdot)$ denotes the variance of the given function, $|u|$
is the cardinality of the set and $0\le p\le 1$, for instance
($p=0.99$). The effective dimension of $f$, in the \emph{truncation
sense} is the smallest integer $d_T$ such that
$\sum_{u\subseteq\{1,2,\dots,d_T\}} \sigma^2(f_u) = p\sigma^2(f)$.
Essentially, the truncation dimension indicates the number of
important variables which predominantly capture the given function
$f$. The superposition dimension takes into account that for some
integrands, the inputs might influence the outcome through their
joint action within small groups.

The PCA decomposition only permits a dimension reduction without
taking into account the particular payoff function of a European
option. In contrast, Imai and Tan~\cite{IT2006} have proposed a
general dimension reduction construction, named Linear
Transformation (LT), that depends on the payoff function and that
minimizes the effective dimension in the truncation sense. They have
shown that the LT approach is more accurate than the standard PCA,
but has a higher computational cost. In a previous paper one of the
authors~\cite{Sab08} has discussed how to implement this technique
quickly and -- with a slower computer -- has obtained computational
times that are about $30$ times smaller that those originally
presented by Imai and Tan \cite{IT2006}.

In the present study we consider time-dependent volatilities: as a
consequence it is not possible to rely on the properties of the
Kronecker product, and the problem is computationally more complex.
In order to simplify the computational complexity, we present a fast
Cholesky (CH) decomposition algorithm tailored for block matrices
that remarkably reduces the computational cost. Moreover, we present
a new path-generation technique based on the Kronecker Product
Approximation (KPA) of the correlation matrix of the
multi-dimensional Brownian path that returns a suboptimal ANOVA
decomposition with a substantial advantage from the computational
point of view.

Our numerical experiment consists in calculating the Randomized QMC
(RQMC) estimation of the prices and the deltas of high-dimensional
Asian basket options in a BS market with time-dependent
volatilities. In order to compute the deltas, we extend to a
multi-assets dependence the procedure employed by Montero and
Kohatsu-Higa~\cite{MKH2003} in a single asset setting. This
procedure is based on the Malliavin Calculus and allows a certain
flexibility that can enhance the localization technique introduced
by Fourni\'e et al. \cite{FLLL1999}. As far as the computation of
Asian options prices is concerned, the KPA and LT approaches are
tested both in terms of accuracy and computational cost. We
demonstrate that the LT construction becomes more efficient than the
PCA even from the computational point of view, provided we use the
CH algorithm that we present and the approach described
in~\cite{Sab08}. The KPA and the PCA constructions perform equally
in terms of accuracy, with the former one requiring a considerably
shorter computational time. However, the KPA and the LT display the
same accuracies in the computation of the deltas. Moreover, we
compare  our simulation experiment -- also using the standard CH and
the PCA decomposition methods -- with pseudo-random and Latin
Hypercube Sampling (LHS) generators.

Remark finally that all the methods described here can accommodate a
market with stochastic volatility where the evolution of the risky
securities is modeled by a mixture of multi-dimensional dynamics as
in the papers by Brigo et al.~\cite{BMR,BMR04}. It is noteworthy to
say that none of these constructions can be applied to Heston-like
multi-dimensional stochastic volatility models. In principle, we
might still use the LT for the Euler discretization of the Heston
model, but this could be no longer applicable with more realistic
schemes that involve discrete random variables as proposed for
instance by Alfonsi~\cite{Al05}.

The paper is organized as follows: Section \ref{sec:Asian} describes
 Asian options. Section \ref{sec:Path} discusses some
path-generation techniques and in particular, presents the fast CH
algorithm and the KPA construction. Section \ref{sec:Num} shows the
numerical tests for the Asian option pricing. Section \ref{sec:Mall}
explains how to represent the deltas of Asian basket options as
expected values with the aid of Malliavin Calculus and shows the
estimated values by RQMC. Section \ref{sec:concl} summarizes the
most important results and concludes the paper.
%%%%%%%%%%%%%%%%%%%%%%%%%%%%%%%%%%%%%%

%%%%%%%%%%%%%%%%%%%%%%%%%%%%%%%%%%%%%%
\section{Asian Basket Options}
\label{sec:Asian} Assume a multi-dimensional BS market with $M$
risky securities and one risk-free asset. Denote $\mathbf{B}\left(
t\right) =\left( B_{1}\left( t\right) ,\dots ,B_{M}\left( t\right)
\right)$  an $M$-dimensional Brownian motion (BM) with correlated
components and $(\mathcal{F}_t)_{t\ge 0}$ the filtration generated
by this BM. Moreover, denote $\rho_{ik}$  the
constant instantaneous correlation between $B_{i}(t)$ and
$B_{k}(t)$, $S_{i}\left( t\right)$ the $i$-th asset price at time
$t$, $\sigma _{i}\left( t\right)$ the instantaneous time-dependent
volatility of the $i$-th asset return and $r$ the continuously
compounded risk-free rate. In the risk-neutral probability, we
assume that the dynamics of the risky assets are
\begin{equation}\label{eq:sec1:1}
dS_{i}\left( t\right) =rS_{i}\left( t\right) dt+\sigma _{i}\left(
t\right)S_{i}\left( t\right)\,dB_{i}\left( t\right) ,\qquad
i=1,\dots ,M.
\end{equation}%
\noindent The solution of Equation (\ref{eq:sec1:1}) is
\begin{equation}
S_{i}\left( t\right) =S_{i}\left( 0\right) \exp\left[
\int_0^t\left( r- \frac{\sigma _i^2\left( s\right)}{2}\right)ds
+\int_0^t\sigma _i\left( s\right)dB_{i}\left( s\right) \right]
,\quad i=1,...,M.\label{eq:sec1:2}
\end{equation}
Discretely monitored Asian basket options are derivative contracts
that depend on the arithmetic mean of the prices assumed by a linear
combination of the underlying securities at precise times
$t_1<t_2\dots<t_N=T$, where $T$ is the maturity of the contract. By
the risk-neutral pricing formula (see for instance Lamberton and
Lapeyre~\cite{LL96}) the fair price at time $t$ of the contract is
\begin{equation}\label{eq:sec1:3}
a\left( t\right) =e^{r(T-t)}\mathbb{E}\left[\left(
\sum_{i=1}^{M}\sum_{j=1}^{N}w_{ij}\,S_{i}\left( t_{j}\right)
-K\right)^+\Bigg|\mathcal{F}_t\right],
\end{equation}%
\noindent with the assumption that $\sum_{i,j}w_{ij}=1$.

Pricing Asian options by simulation hence requires the discrete
averaging of the solution (\ref{eq:sec1:2}) at a finite set of times
$\{t_1,\dots,t_N\}$. This sampling procedure yields
\begin{equation}\label{1.2.3}
S_i(t_j) = S_i(0)\exp\bigg[\int_0^{t_j}\left(r -
\frac{\sigma^2_i(t)}{2}\right)dt + Z_i(t_j)\bigg]\quad i=1,\dots,M,
j=1,\dots, N,
\end{equation}
\noindent where the components of the vector
$$\left(Z_1(t_1),\ldots,
Z_1(t_N);\,Z_2(t_1),\ldots,Z_2(t_N);\,\ldots;\,Z_M(t_1),\ldots,Z_M(t_N)\right)^T$$
are $M\times N$ normal random variables with zero mean and the
following covariance matrix
\begin{equation}\label{Cov1}
\Sigma_{MN} = \left( \begin{array}{cccc} \Sigma(t_1) & \Sigma(t_1) &
\ldots & \Sigma(t_1)
\\ \Sigma(t_{1}) & \Sigma(t_2)& \ldots & \Sigma(t_2)
\\ \vdots & \vdots & \ddots & \vdots
\\\Sigma(t_{1}) & \Sigma(t_2) & \ldots & \Sigma(t_{N})
\end{array} \right),
\end{equation}
\noindent where the elements of the $M\times M$ submatrices
$\Sigma(t_n)$ are $\left({\Sigma}(t_n)\right)_{ik} =\int_0^{t_n}\rho
_{ik}\sigma _{i}(s)\sigma _{k}(s)ds$ with $i,k=1,\ldots,M;\,
n=1,\ldots,N$. This setting would allow time-dependent correlations
as well. In the case of constant volatilities the covariance matrix
is
\begin{equation}
\Sigma_{MN} = \left( \begin{array}{cccc} t_{1}\Sigma & t_{1}\Sigma &
\ldots & t_{1}\Sigma
\\ t_{1}\Sigma & t_2\Sigma & \ldots & t_{2}\Sigma
\\ \vdots & \vdots & \ddots & \vdots
\\t_{1}\Sigma & t_2\Sigma & \ldots & t_{N}\Sigma
\end{array} \right)\label{Cov2},
\end{equation}
\noindent where now $\Sigma$ denotes the $M\times M$ covariance
matrix of the logarithmic returns of the assets. It follows from the
last equation that the covariance matrix $\Sigma_{MN}$ can be
represented as $R\otimes\Sigma$, where $\otimes$ denotes the
Kronecker product and $R$ is the auto-covariance matrix of a single
BM. This simplification is not possible in the case of
time-dependent volatilities. We recall that the elements of $R$ are
\begin{equation}\label{brown}
R_{ln}=t_l\wedge t_n,\quad l,n=1,\dots N.
\end{equation}
\noindent $R$ has the peculiarity to be invariant for a reflection
about the diagonal.
\begin{stdef}[Boomerang Matrix] Let $B\in\mathbb{R}^{n_B\times n_B}$ be a square matrix and let
$\mathbf{b}=\left(b_1,\dots,b_{n_B}\right)\in\mathbb{R}^{n_B}$. $B$
is a boomerang matrix if
\begin{equation}\label{eq:boom}
B_{hp}=b_{h\wedge p},\quad h,p=1,\dots,n_B.
\end{equation}
We call $\mathbf{b}$ the elementary vector associated to $B$.
\end{stdef}

\noindent As a consequence $R$ is boomerang, and in general the
auto-covariance matrix of every Gaussian process is boomerang. This
definition can also be extended to block matrices as follows.
\begin{stdef}[Block Boomerang Matrix]
Partition the rows and the columns of a square matrix
$B\in\mathbb{R}^{n_B\times n_B}$ to obtain:
\begin{equation}
B = \left( \begin{array}{ccc} B_{11} & \ldots & B_{1P}
\\ \vdots &  \ddots & \vdots
\\B_{P1} & \ldots & B_{PP}
\end{array} \right),
\end{equation}
\noindent where for $h,p=1,\dots,P$, $B_{hp}\in\mathbb{R}^{D\times
D}$ designates the $(h,p)$ square submatrix and $n_B=P\times D$.
Given $P$ matrices $B_1,\dots,B_P$ with $B_h\in\mathbb{R}^{D\times
D},h=1,\dots,P$, $B$ is a boomerang block  matrix if
\begin{equation}\label{eq:boomBlock}
B_{hp}=B_{h\wedge p},\quad h,p=1,\dots,n_B.
\end{equation}
We call $\mathbf{b}=\left(B_1,\dots,B_P\right)^T$ the elementary
block vector associated to $B$.
\end{stdef}

\noindent From these definitions we have that $\Sigma_{MN}$ is block
boomerang.

The payoff at maturity of the Asian basket option now is $a(T) =
\left(g(\mathbf{Z})-K\right)^+$ with
\begin{equation}
g(\mathbf{Z})=\sum_{k=1}^{M\times N} \exp\left (\mu_k + Z_k
\right)\label{1.2.6}
\end{equation}
where $\mathbf{Z}\sim\mathcal{N}(0,\Sigma_{MN})$ and
\begin{equation}
\mu_k = \ln(w_{k_1k_2}S_{k_1}(0)) + rt_{k_2}-\int_0^{t_{k_2}}
\frac{\sigma_{k_1}^2(t)}{2}\,dt
\end{equation}\label{1.2.8}

\noindent with $k_1=(k-1)\bmod M;\, k_2 = \lfloor(k-1)/M\rfloor+1;\,
k=1,\dots,M$, where $\lfloor x\rfloor$ denotes the greatest integer
less than or equal to $x$.
%%%%%%%%%%%%%%%%%%%%%%%%%%%%%%%%%%%%%%

%%%%%%%%%%%%%%%%%%%%%%%%%%%%%%%%%%%%%%
\section{Path-generation Techniques}\label{sec:Path}
From the previous discussion it comes out that the pricing of Asian
basket options by simulation requires an averaging on the sample
trajectories of an $M$-dimensional BM. In general, if
$\mathbf{Y}\sim \mathcal{N}(0,\Sigma_Y)$ and
$\mathbf{X}\sim\mathcal{N}(0,I)$ are two $N$-dimensional Gaussian
random vectors, we will alawys be able to write
$\mathbf{Y}=C\mathbf{X}$, where $C$ is a matrix such that:
\begin{equation}
    \Sigma_Y=C C^T.\label{prob}
\end{equation}
and the core problem consists in finding the matrix $C$. In our case
$\Sigma_Y$ coincides with $\Sigma_{MN}$ of Equation (\ref{Cov1}).
The accuracy of the standard MC method does not depend on the choice
of the matrix $C$ because the order of the random variables is not
important. However, a choice of $C$ that reduces the nominal
dimension would improve the efficiency of the (R)QMC. In the
following we discuss some possibilities.
%%%%%%%%%%%%%%%%%%%%%%%%%%%%%%%%%%%%%%%%%%%%%%%%%%%%%%
%%%%%%%%%%%%%%%%% CHOLESKY
%%%%%%%%%%%%%%%%%%%%%%%%%%%%%%%%%%%%%%%%%%%%%%%%%%%%%%%%%%%%%%%%%%%55
\subsection{Cholesky Construction}\label{sec:Chol}
The CH decomposition simply finds the matrix $C$ among all the
\emph{lower triangular matrices}. In the case of constant
volatilities the matrix $\Sigma_{MN}$ is the Kronecker product of
$R$ and $\Sigma$, and the Kronecker product shows compatibility with
the CH decomposition (see Pitsianis and Van Loan~\cite{PV1993}).
Indeed, denote $C_{\Sigma_{MN}}$, $C_{R}$ and $C_{\Sigma}$ the CH
matrices associated to $\Sigma_{MN}$, $R$ and $\Sigma$ respectively;
we then have
\begin{equation}\label{CholKron}
    C_{\Sigma_{MN}}=C_{R}\otimes C_{\Sigma}.
\end{equation}
This now allows a remarkable reduction of the computational cost: it
turns out in fact that a $O\left((M\times N)^3\right)$ computation
is reduced to a $O\left(M^3\right)+O\left(N^3\right)$ one.

When time-dependent volatilities are considered, however, we can no
longer use these properties of the Kronecker product.  In any case,
since $\Sigma_{MN}$ is a block boomerang  matrix, we can use the
following result:
\begin{prop}\label{CholBlocInv} Let $B\in\mathbb{R}^{n_B\times n_B}$  be a block boomerang matrix
and let $\left(B_1,\dots,B_P\right)^T$, where
$B_h\in\mathbb{R}^{D\times D},h=1,\dots,P$ with $n_B=P\times D$, be
its associated elementary block vector. $C_B$, the CH matrix
associated to B, is given by:
\begin{equation}
C_{B}= \left( \begin{array}{cccc} C_{1}& 0 & \ldots &0\\
\vdots & C_2 & \ddots &0\\
\vdots & \vdots & \ddots & \vdots
\\C_{1} & C_2 & \ldots & C_{P}\
\end{array} \right)\label{Chol2}
\end{equation}
where the $D\times D$ blocks $C_h$, $h=1,\dots,P$ are
\begin{equation}
C_h = \mathrm{Chol}\left(B_h-B_{h-1}\right)
\end{equation}
with $\mathrm{Chol}$ denoting the CH factorization, and we assume
$B_0=0$.
\end{prop}
\begin{proof}
Consider the $h^{\mathrm{th}}$ row of $C_B$ and the
$m^{\mathrm{th}}$ row of its transposed matrix; we then have
\begin{eqnarray*}
\left(C_1,\dots,C_h,0,\dots,0\right)^T \cdot
\left(C_1^T,\dots,C_m^T,0,\dots,0\right)^T&=&\sum_{l=1}^{h\wedge
m}C_lC_l^T\\
&=&\sum_{l=1}^{h\wedge m}(B_l-B_{l-1})=B_{h\wedge m}
\end{eqnarray*}
\noindent and this concludes the proof.
\end{proof}
%%%%%%%%%%%%%%%%%%%%%%%%%%%%%%%%%%%%%%%%%%%%%%%%%%%%%%
%%%%%%%%%%%%%%%%% PCA
%%%%%%%%%%%%%%%%%%%%%%%%%%%%%%%%%%%%%%%%%%%%%%%%%%%%%%%%%%%%%%%%%%%55
\subsection{Principal Component Analysis}
Acworth et al.~\cite{ABG1998} have proposed a path generation
technique based on the PCA. Following this approach we consider the
spectral decomposition of $\Sigma_{MN}$
\begin{equation}\label{5.2.1}
   \Sigma_{MN} = E\Lambda E^T = (E \Lambda^{1/2}) (E
   \Lambda^{1/2})^T,
\end{equation}
\noindent where $\Lambda$ is the diagonal matrix of all the positive
eigenvalues of $\Sigma_{MN}$ sorted in decreasing order and $E$ is
the orthogonal matrix ($EE^T=I$) of all the associated eigenvectors.
The matrix $C$ solving Equation (\ref{prob}) is then
$E\Lambda^{1/2}$. The amount of variance explained by the first $k$
principal components is the ratio: $\frac{\sum_{i=1}^k
\lambda_i}{\sum_{i=1}^d \lambda_i}$ where $d$ is the rank of
$\Sigma_{MN}$. The PCA construction permits the statistical ranking
of the normal factors, while this is not possible by the CH
decomposition. For the market with constant volatilities, the
Kronecker product reduces this calculation into the computation of
the eigenvalues and vectors of the two smaller matrices $R$ and
$\Sigma$. All these simplifications are no longer valid for the
time-dependent volatilities. However, we can reduce the
computational cost for the PCA decomposition in the following way.

Take $M_1,M_2,M_3$ and $M_4$, respectively $p\times p, p\times q,
q\times p$ and $q\times q$ matrices, and suppose that $M_1$ and
$M_4$ are invertible. Assume
\begin{equation*}
M = \left(\begin{array}{cc} M_1 & M_2\\ M_3 & M_4
\end{array}\right)
\end{equation*}
and define $S_1=M_4-M_3M_1^{-1}M_2$ and $S_4=M_1-M_2M_4^{-1}M_3$,
the Schur complements of $M_1$ and $M_4$, respectively. Then by
Schur's lemma the inverse $M^{-1}$ is:
\begin{equation}
M^{-1} = \left(\begin{array}{cc} S_4 & -M_1^{-1}M_2S_1^{-1}\\
-M_4^{-1}S_4^{-1} & S_1^{-1}
\end{array}\right).
\end{equation}
Taking into account the previous result it is possible to prove the
following proposition
\begin{prop}\label{PCAInv}
Let $B\in\mathbb{R}^{n_B\times n_B}$  be a block boomerang matrix
and let $\left(B_1,\dots,B_P\right)^T$, where
$B_h\in\mathbb{R}^{D\times D},h=1,\dots,P$ with $n_B=P\times D$, be
its associated elementary block vector. The inverse of $B$ is
symmetric block tri-diagonal. The blocks on the lower (and upper)
diagonal are $T_{l}=-\left(B_{l+1}-B_l\right)^{-1}$, $l=1,\dots,P-1$
while those on the diagonal are
$D_{m}=\left(B_m-B_{m-1}\right)^{-1}\left(B_{m+1}-B_{m-1}\right)\left(B_{m+1}-B_{m}\right)^{-1}$,
$m=1,\dots,P$, with the assumption that $B_{0}=B_{N+1}=0$:
\begin{equation}\label{in}
    \begin{array}{c}
    B^{-1} =
    \left(
    \begin{array}{ccccc}
        D_1 & T_1 & 0 & \ldots & 0
        \\
        T_1 & D_2 & T_2 & \ddots & \vdots
        \\
        0 & T_2 &  \ddots & \ddots & 0
        \\
        \vdots & \vdots & \ddots & \ddots & T_{P-1}
        \\ 0 & 0 & 0 & T_{P-1}& D_P
        \end{array}
    \right)
    \end{array}
\end{equation}

%\begin{equation}\label{in}
%\begin{array}{c}
%  B^{-1} =
%  \left( \begin{array}{cccccc}
%D_1 & T_1 & 0 & \ldots & \ldots & 0
%\\ T_1 & D_2 & T_2 & 0 &
%\ldots & \vdots
%\\ 0 & T_2 &
%D_3 & T_3 & \vdots
%\\ \vdots & 0 &  T_3 &
% \ddots & \ddots & \vdots
%\\ \vdots & \vdots & \ddots & \ddots & D_{P-1} &
%T_{P-1}
%\\ 0 & 0 & 0 & 0 & T_{P-1}& D_P
%\end{array} \right)
%\end{array}
%\end{equation}
\end{prop}
\noindent This property can be used to reduce the computational cost
of evaluating the PCA decomposition in the case of time-dependent
volatilities and in general for multi-dimensional Gaussian
processes. Indeed, if $B$ is a non-singular square matrix then the
eigenvalues of the $B^{-1}$ are the reciprocal of the eigenvalues of
$B$ and the eigenvectors coincide.
%%%%%%%%%%%%%%%%%%%%%%%%%%%%%%%%%%%%%%%%%%%%%%%%%%%%%%
%%%%%%%%%%%%%%%%% LT
%%%%%%%%%%%%%%%%%%%%%%%%%%%%%%%%%%%%%%%%%%%%%%%%%%%%%%%%%%%%%%%%%%%55
\subsection{Linear Transformation}
Imai and Tan~\cite{IT2006} have considered the following class of LT
as a solution of (\ref{prob}):
\begin{equation}\label{4.3.1}
    C^{\mathrm{LT}} = C^{\mathrm{Ch}} A
\end{equation}
\noindent where $C^{\mathrm{Ch}}$ is the CH matrix associated to the
covariance matrix of the normal random vector to be generated, and
$A$ is an orthogonal matrix, i.e. $AA^T=I$. The matrix $A$ is
introduced with the main purpose of minimizing the effective
dimension of a simulation problem in the truncation sense. Imai and
Tan~\cite{IT2006} have proposed to approximate an arbitrary function
$g$, such that $(g-K)^+$ is the payoff function of a European
derivative contract, with its first order Taylor expansion around
$\hat{\epsilon}$
\begin{equation}\label{4.3.12}
    g(\mathbf{\epsilon}) = g(\mathbf{\hat{\epsilon}}) +
    \sum_{l=1}^n\frac{\partial
    g}{\partial\epsilon_l}\Big|_{\mathbf{\epsilon}=\mathbf{\hat{\epsilon}}}\Delta\epsilon_l.
\end{equation}
The approximated function is linear in the standard normal random
vector $\Delta\epsilon$. Considering an arbitrary point about which
we form the expansion, such as $\mathbf{\hat{\epsilon}=0}$, we can
derive the first column of the optimal orthogonal matrix $A^*$. It
is possible to find the complete matrix by expanding $g$ about
different points and then compute the optimization algorithm. Imai
and Tan~\cite{IT2006} have set:
$\hat{\epsilon}_1=\mathbf{0}=(0,0,\dots,0),\,\hat{\epsilon}_2=(1,0,\dots,0),\dots,\,
\hat{\epsilon}_n=(1,\dots,1,0)$, where the $k$-th point has $k-1$
non-zero components. The optimization can then be formulated as
follows:
\begin{equation}\label{4.3.13}
    \max_{\mathbf{A_{\cdot k}}\in \mathbf{R^n}}\left(\frac{\partial g}
    {\partial\epsilon_k}\Big|_{\mathbf{\epsilon}=\mathbf{\hat{\epsilon}_k}}\right)^2,\quad
    k=1,\dots,n,
\end{equation}
subject to $\|\mathbf{A_{\cdot k}}\|=1$ and $\mathbf{A^*_{\cdot
j}}\cdot \mathbf{A_ {\cdot k}}=0;\, j=1,\dots,k-1;\, k\le n$. In the
case of Asian basket options we have
\begin{equation}\label{4.4.4}
    g(\mathbf{\epsilon}) = g(\mathbf{\hat{\epsilon}}) +
    \sum_{l=1}^{NM}\left[\sum_{i=1}^{NM}\exp\left(\mu_i+\sum_{k=1}^{NM}C_{ik}\hat{\epsilon}_k\right)C_{il}\right]\Delta\epsilon_l.
\end{equation}
Imai and Tan~\cite{IT2006} have proved the following result:
\begin{prop} Consider an Asian basket options in a BS model,
define:
\begin{eqnarray}
  \mathbf{d^{(p)}} &=&
  \left(e^{\left(\mu_1+\sum_{k=1}^{p-1}C_{1k}^*\right)},\dots,e^{\left(\mu_{MN}+\sum_{k=1}^{p-1}C_{MN,k}^*\right)}\right)^T
  \\
  \mathbf{B^{(p)}} &=& \left(C^{\mathrm{Ch}}\right)^T\mathbf{(d^{(p)})},\quad p=1,\dots,MN.
\end{eqnarray}
\noindent Then the $p$-th column of the optimal matrix $A^*$ is
\begin{equation}
\mathbf{A_{\cdot
p}^*}=\pm\frac{\mathbf{B^{(p)}}}{\|\mathbf{B^{(p)}}\|}\quad
p=1,\dots,MN.
\end{equation}
\end{prop}
\noindent The matrices $C_{ik}^*$, $k<p$ have been already found at
the $p-1$ previous steps. $\mathbf{A_{\cdot p}}$ must be orthogonal
to all the other columns. This feature can be easily obtained by an
incremental QR decomposition as described in Sabino~\cite{Sab08}.
%%%%%%%%%%%%%%%%%%%%%%%%%%%%%%%%%%%%%%%%%%%%%%%%%%%%%%
%%%%%%%%%%%%%%%%% KPA
%%%%%%%%%%%%%%%%%%%%%%%%%%%%%%%%%%%%%%%%%%%%%%%%%%%%%%%%%%%%%%%%%%%55
\subsection{Kronecker Product Approximation}
In a time-dependent volatilities BS market the covariance matrix
$\Sigma_{MN}$ has time-dependent blocks. The multi-dimensional BM is
the unique source of risk in the BS market and the generation of the
trajectories of the $1$-dimensional BM does depend on the
volatilities. Based on these considerations, we propose to find a
constant covariance matrix among the assets $H$, in order to
approximate, in an appropriate sense, the matrix $\Sigma_{MN}$ as a
Kronecker product of $R$ and $H$.  In the following we illustrate
the proposed procedure that we label Kronecker Product Approximation
(KPA). Pitsianis and Van Loan \cite{PV1993} have proved the
following proposition.
\begin{prop}
Suppose $G\in\mathbb{R}^{m\times n}$ and
$G_1\in\mathbb{R}^{m_1\times n_1}$ with $m=m_1m_2$ and $n=n_1n_2$.
Consider the problem of finding $G_2^*\in\mathbb{R}^{m_1\times n_1}$
that realizes the minimum
\begin{equation}\label{5.4.2}
\min_{G_2\in\mathbb{R}^{m_1\times n_1}}\parallel G - G_1\otimes
G_2\parallel^2_F,
\end{equation}
\noindent where $\parallel\cdot\parallel^2_F$ denotes the Frobenius
norm. For fixed $h=1,\dots,m_2$ and $l=1,\dots,n_2$ denote
$\mathcal{R}(G)_{hl}$ the $m_1\times n_1$ matrix defined by the rows
$h,h+m_2,h+2m_2,\dots,h+(m_1-1)m_2$ and the columns
$l,l+n_2,l+2n_2,\dots,l+(n_1-1)n_2$ of the original matrix $G$. The
elements of $G_2^*$ which gives (\ref{5.4.2}) are
\begin{equation}\label{5.4.3}
(G_2^*)_{hl} =
\frac{\mathrm{Tr}\left(\mathcal{R}(G)^T_{hl}G_1\right)}{\mathrm{Tr}\left(G_1G_1^T\right)}\quad
h=1,\dots,m_2,\, l=1,\dots,n_2,
\end{equation}
where ${\mathrm{Tr}}$ indicates the trace of a matrix.
\end{prop}
\noindent In our setting, we have $G=\Sigma_{MN}$, $G_1=R$ and
$G_2=H$. We note that for any $i,j=1,\dots,N$,
$\mathcal{R}(\Sigma_{MN}))_{ij}$ is a $N\times N$ boomerang matrix.
Moreover, given two general $N\times N$ boomerang matrices $A$ and
$B$, by direct computations we can prove
\begin{equation}
\mathrm{Tr}(A^TB) = \mathrm{Tr}(AB) = \sum_{j=1}^N
\left(2(N-j)+1\right)a_{jj}b_{jj}\label{5.4.5}.
\end{equation}
Then we perform the PCA decomposition of $R\otimes H$ relying on the
properties of the Kronecker product. However, if we use the PCA
decomposition of the matrix $F = R\otimes{H}$ we do not get the
required path. In order to produce the required trajectory we take
\begin{equation}\label{eq:KPA}
  \mathbf{Z}=C^{\mathrm{KPA}}\mathbf{\epsilon}=
  C_{\Sigma_{MN}}(C_F)^{-1}E_H\Lambda_H^{1/2}\mathbf{\epsilon}
\end{equation}
where $C_{\Sigma_{MN}}$ and $C_F$ are the CH matrices associated to
$\Sigma_{MN}$ and $F$, respectively, and $E_H\Lambda_H^{1/2}$ is the
PCA decomposition of $F$. The matrix $C^{\mathrm{KPA}}$ produces the
correct covariance matrix; indeed, denoting $P=E_H\Lambda_H^{1/2}$,
we have
\begin{equation*}
  C^{\mathrm{KPA}}\left(C^{\mathrm{KPA}}\right)^T =
  C_{\Sigma_{MN}}(C_F)^{-1}PP^T\left[(C_F)^{-1}\right]^TC_{\Sigma_{MN}}^T =
  C_{\Sigma_{MN}}C_{\Sigma_{MN}}^T=\Sigma_{MN}
\end{equation*}
because $PP^T = C_FC_F^T=F$. Our fundamental assumption is that the
principal components of $\mathbf{Z}$ are not so different from those
of the normal random vector $\mathbf{Z}'$ whose covariance matrix is
$F$. We expect that the KPA decomposition would produce an effective
dimension higher than the effective dimension obtained by the PCA
decomposition, but with an advantage from the computational point of
view. Due to properties of the Kronecker product, Equation
(\ref{eq:KPA}) becomes
\begin{equation}\label{5.4.9}
\mathbf{Z} = C_{\Sigma_{MN}}\left(C_R^{-1}\otimes
C_H^{-1}\right)E_H\Lambda_H^{1/2}\mathbf{\epsilon},
\end{equation}
where $C_R$ and $C_H$ are the CH matrices of $R$ and $H$,
respectively. This matrix multiplication can be carried out quickly
by block-matrices multiplication and knowing that, due to the
Propositions \ref{CholBlocInv} and \ref{PCAInv}, $C_R^{-1}$ is a
sparse bi-diagonal matrix.
%%%%%%%%%%%%%%%%%%%%%%%%%%%%%%%%%%%%%%
%
%%%%%%%%%%%%%%%%%%%%%%%%%%%%%%%%%%%%%%
\section{Computing the Option Price}\label{sec:Num}
We will now estimate the fair price of an Asian option on a basket
of $M=10$ underlying assets with $N=250$ sampled points in the BS
model with time-dependent volatilities having the following
expression
\begin{equation}
\sigma_i(t) =
\hat{\sigma}_i(0)\,\exp\left(-t/\tau_i\right)+\sigma_i(+\infty),\quad
i=1,\dots M.
\end{equation}
The parameters chosen for the simulation are listed in Table
\ref{parameters} ($\hat{\sigma}_i(0)$ is then
$\sigma_i(0)-\sigma_i(+\infty)$).
 \begin{table}\centering
 \caption{Inputs Parameters}\label{parameters}
 \begin{tabular}{cc}
     \begin{tabular}{lll}
     \\
         \hline
         $S_{i}(0)$&=&$100, \quad\forall i=1\dots,N$\\
         $K$&$\subset$&$\{90,100,110\}$\\
         $r$ &=& $4$\%\\
         $T$&=&$1$\\
%         \hline
%     \end{tabular}
%
%     \begin{tabular}{lll}
%         \hline
         $\sigma_i(0)$ &=&$10\%+\frac{i-1}{9}40\% \quad i=1\dots,N$\\
         $\sigma_i(+\infty)$&=&$ 9\% \quad \forall i=1\dots,N $\\
         $\tau_i$&=&$1.5 \quad\forall i=1\dots,N$\\
         $\rho_{ij}$&$\subset$&$\{0, 40\} \quad i,j=1\dots,N$\\
         \hline
     \end{tabular}
 \end{tabular}
 \end{table}
 \begin{table}\centering%
 \caption{Effective Dimensions. Time-dependent  Volatilities}\label{Eff_Dim}%
 \begin{tabular}{cc}
     \begin{tabular}{l}
     \\
         $\rho = 0\%$\\
         \begin{tabular}{llll}
             \hline
             Ch & PCA & LT & KPA \\
             $d_T > 1900$  & $d_T = 14$ & $d_T = 10$ & $d_T = 19$ \\
             \hline
             \end{tabular}
     \end{tabular}
 \\
     \begin{tabular}{l}
     \\
         $\rho = 40\%$\\
         \begin{tabular}{llll}
             \hline
             Ch & PCA & LT & KPA \\
             $d_T > 1900$  & $d_T = 9$ & $d_T = 8$ & $d_T = 11$ \\
             \hline
         \end{tabular}\\
     \end{tabular}
 \end{tabular}
 \end{table}
We implement the numerical investigation in two parts: first we test
the effectiveness of the path-generation constructions on dimension
reduction and compute their computational times, and then we compare
the accuracy of the simulation.

The Table \ref{Eff_Dim} shows the effective dimensions obtained by
all the path-generation methods ($p=0.99$). The LT construction
provides the lowest effective dimension, while the PCA decomposition
performs almost as well as the LT approach for the correlation case
only, and the KPA returns a slightly higher effective dimension. The
CH decomposition collects $98.58\%$ and $98.70\%$ of the total
variance for $d_T\approx 2000$ for the uncorrelated and correlated
cases, respectively. To have a more accurate comparison, Table
\ref{ElTimes} displays the elapsed times computed in
Sabino~\cite{Sab08} using an \emph{ad hoc} incremental QR algorithm
for the LT and assuming constant volatilities each equal to
$\sigma_i(0)$ of Table \ref{parameters}. The computation is
implemented in MATLAB running on a laptop with an Intel Pentium M,
processor 1.60 GHz and 1 GB of RAM. We compute $50$ optimal columns
for the LT technique.
 \begin{table}\centering
 \caption{Computational Times in Seconds}\label{ElTimes}
 \begin{tabular}{cc}
 \\
 Constant Volatilities
 \\
 \\
     \begin{tabular}{ll}
     $\rho = 0\%$ & $\rho = 40\%$\\
         \begin{tabular}{lll}
             \hline
             Ch & PCA & LT  \\
             $0.60$  & $25.77$ & $71.14$ \\
             \hline
         \end{tabular}
         &
         \begin{tabular}{lll}
             \hline
             Ch & PCA & LT  \\
             $0.59$  & $25.55$ & $71.02$ \\
             \hline
         \end{tabular}
     \end{tabular}
     \\
     \\
Time-dependent  Volatilities
     \\
     \\
     \begin{tabular}{ll}
         $\rho = 0\%$ & $\rho = 40\%$\\
         \begin{tabular}{llll}
             \hline
             Ch & PCA & LT & KPA \\
             $0.62$  & $565.77$ & $71.65$ & $28.25$ \\
             \hline
         \end{tabular}
         &
         \begin{tabular}{llll}
             \hline
             Ch & PCA & LT & KPA \\
             $0.62$  & $568.55$ & $71.20$ & $28.33$ \\
             \hline
         \end{tabular}
     \end{tabular}
     \\
 \end{tabular}
 \end{table}
The CH algorithm for block boomerang  matrices has almost the same
cost as the one relying on the properties of the Kronecker product.
As a consequence, the LT also requires almost the same computational
cost, while the PCA needs a time almost $20$ times higher because we
can no longer rely on the properties of the Kronecker product. In
contrast, the KPA has almost the same computational time as the PCA
in the constant volatility case and is the best performing
path-generation technique from the computational time point of view.
We have applied Proposition \ref{PCAInv} to implement the PCA and
computed the eigenvalues and eigenvectors of $\Sigma_{MN}$ relying
only on the \textit{sparse} function of MATLAB. It is noteworthy to
say that there exist algorithms tailored for the computation of the
eigenvalues and eigenvectors of tri-diagonal symmetric block
matrices that can further reduce the computational time.
 \begin{table}[tbp]\centering
 \caption{Estimated At-the Money Prices and
 Errors.}\label{PriceErr}
     \begin{tabular}{ccc}
     \\
         & $\rho = 0$ & $\rho = 40\%$\\
         \hline
         MC
         &
         \begin{tabular}{|ccc}
             & Price & RMSE\\
             Ch & $3.18200$&$0.01300$\\
             KPA & $3.12400$&$0.01300$\\
             PCA & $3.10600$&$0.01300$\\
             LT & $3.11100$&$0.01300$\\
         \end{tabular}
         &
         \begin{tabular}{cc}
             Price & RMSE\\
             $5.18900$&$0.02600$\\
             $5.19400$&$0.02600$\\
             $5.20100$&$0.02600$\\
             $5.24200$&$0.02600$\\
         \end{tabular}\\
         \hline
         LHS
         &
         \begin{tabular}{|ccc}
             & Price & RMSE\\
           Ch&  $3.12200$&$0.00750$\\
            KPA & $3.12440$&$0.00550$\\
             PCA & $3.12010$&$0.00540$\\
             LT & $3.12200$&$0.00290$\\
         \end{tabular}
         &
         \begin{tabular}{cc}
             Price & RMSE\\
             $5.20000$&$0.01200$\\
             $5.20950$&$0.00340$\\
             $5.20070$&$0.00320$\\
             $5.20090$&$0.00120$\\
         \end{tabular}\\
         \hline
         RQMC
         &
         \begin{tabular}{|ccc}
             & Price & RMSE\\
          Ch&   $3.11240$&$0.00550$\\
          KPA&   $3.12240$&$0.00086$\\
          PCA&   $3.12140$&$0.00078$\\
          LT &   $3.12230$&$0.00021$\\
         \end{tabular}
         &
         \begin{tabular}{cc}
             Price & RMSE\\
             $5.19500$&$0.01500$\\
             $5.20080$&$0.00054$\\
             $5.20080$&$0.00069$\\
             $5.20080$&$0.00019$\\
         \end{tabular}\\
         \hline
     \end{tabular}
 \end{table}

In the second part of our investigation we launch a simulation in
order to estimate the Asian option price using 10 replications each
of 8192 random points following the strategy in Imai and Tan
\cite{IT2006}. We use different random generators: standard MC, LHS
and RQMC generators. Concerning the computational times of the price
estimation, the CPU ratio between LHS and RQMC is almost $1$ while
standard MC is $1.33$ faster. Moreover, the LT construction needs a
time that is almost $1/30$ of the total computational time of the
LHS or RQMC simulation. As a RQMC generator we use a Matou\^{s}ek
scrambled version (see Matou\^{s}ek~\cite{Ma1998}) of the
$50$-dimensional Sobol´ sequence satisfying Sobol's property A (see
Sobol~\cite{Sobol76}). We pad the remaining random components out
with LHS. This hybrid strategy is intended to investigate the
effective improvement of the decomposition methods when coupled with
(R)QMC. Indeed, it can be proven that LHS gives good variance
reductions when the target function is the sum of one-dimensional
functions (see Stein~\cite{Stein87}). On the other hand, the LT
method is conceived to capture the lower effective dimension in the
truncation sense for linear combinations. As a consequence, we
should already observe a high accuracy when running the simulation
using LHS combined with LT. We expect the KPA technique to produce a
suboptimal decomposition in the sense of ANOVA, with the advantage
of a lower computational effort. Our setting is thought to test how
large is the improvement given by all the factorizations. Tables
\ref{PriceErr} and \ref{DeltaErr} present the results of our
investigation. The prices in Table \ref{PriceErr} are all in
statistical agreement. Those obtained with the CH decomposition are
almost not sensitive to the random number generator. KPA, PCA and LT
all provide good improvements both for the LHS and RQMC
implementations for all the strike prices. The LT has an evident
advantage compared to the PCA and KPA constructions in the
uncorrelated case. In contrast, we observe that the KPA and
PCA-based simulations give almost the same accuracy, both assuming
uncorrelated and correlated asset returns. Considering the total
computational cost and accuracy we observe that the KPA performs
better than the standard PCA. Moreover, all these constructions can
be employed in stochastic and local volatility models that are based
on the mixture of multi-dimensional dynamics for basket options as
done in Brigo et al.~\cite{BMR}.
%%%%%%%%%%%%%%%%%%%%%%%%%%%%%%%%%%%%5555
%%%%   COMPUTING THE SENSITIVITIES
%%%%%%%%%%%%%%%%%%%%%%%%%%%%%%%%%%%%5555
\section{Computing the Sensitivities}\label{sec:Mall}
In the financial jargon, a Greek is the derivative of an option
price with respect to a parameter. A Greek is therefore a measure of
the sensitivity of the price with respect to one of its parameters.
The deltas ($\Delta$'s) are the components of the gradient of the
discounted expected outcome of the option with respect to the
initial values of the assets. The problem of computing the Greeks in
finance has been studied by several authors. In the following we
extend the methodology employed by Montero and Kohatsu-Higa
\cite{MKH2003}, based on the use of Malliavin Calculus, to the
multi-assets case. The main difficulties of this extension lie in
the fact that the assets are now correlated and the formulas in
Montero and Kohatsu-Higa~\cite{MKH2003} cannot be directly extended
to the multi-dimensional case. Moreover, the localization technique,
introduced by Fourni\'e et al~\cite{FLLL1999}, should generally
control all the components of the multi-dimensional BM to improve
the accuracy of the estimation.  We write the dynamics
(\ref{eq:sec1:1}) with respect to a $M$-dimensional BM
$\mathbf{W}(t)$ with uncorrelated components
\begin{equation}
dS_i(t)=rS_i(t)dt +
S_i(t)\sigma_i(t)\sum_{m=1}^M\alpha_{im}(t)dW_m(t)\quad
i=1,\dots,M,
\end{equation}
where $\sum_{m=1}^M\alpha_{im}\alpha_{km}=\rho_{ik}$ and we have defined
$\sigma_{im}(t)=\sigma_i(t)\sum_{m=1}^M\alpha_{im}$.

The Malliavin calculus is a theory of variational stochastic
calculus and provides the mechanics to compute derivatives and
integration by parts of random variables (see Nualart~\cite{Nu06}
for more on Malliavin Calculus).

Denote by $D_s^1,\dots,D_s^M$ the Malliavin derivatives with respect
to the components of $\mathbf{W}(t)$, while
$\delta^{\mathrm{Sk}}=\sum_{m=1}^M\delta_m^{\mathrm{Sk}}$ represents
the Skorohod integral with $\delta_m^{\mathrm{Sk}}$ indicating the
Skorohod integral on the single $W_m(t)$. The domains of the
Malliavin derivatives and the Skorohod integral are denoted by
$\mathbb{D}^{1,2}$ and $\mathrm{dom}(\delta^{\mathrm{Sk}})$,
respectively, while $\delta^{\mathrm{Kr}}$ indicates the Kronecker
delta. We prove the following proposition.
\begin{prop}\label{prop:1}
Denote  $\mathbf{x}=\mathbf{S}(0)$,  and $G_k$ the partial
derivative
\begin{equation}
G_k=\frac{\partial m(T)}{\partial x_k}=\frac{\sum_{j=1}^N
w_{kj}S_k(t_j)}{x_k},\quad k=1,\dots,M,
\end{equation}
\noindent where $m(T)=\sum_{j=1}^N w_{kj}S_k(t_j)$. Knowing that
$a(T)\in\mathbb{D}^{1,2}$,  the $k$-th delta (the $k$-th component
of the gradient) is
\begin{equation}
\Delta_k =\frac{\partial a(0)}{\partial x_k}=e^{-rT}
\mathbb{E}\left[a'(T)G_k\right]=e^{-rT}\mathbb{E}\left[a(T)\sum_{m=1}^M\delta_m^{\mathrm{Sk}}(G_ku_m)\right],
\end{equation}
\noindent where $\mathbf{u}=(u_1,\dots,u_M)\in
\mathrm{dom}(\delta^{\mathrm{Sk}})$, $\mathbf{z}=(z_1,\dots,z_m)\in
\mathrm{dom}(\delta^{\mathrm{Sk}})$, $G_k\mathbf{u}\in
\mathrm{dom}(\delta^{\mathrm{Sk}})$ and
\begin{eqnarray*}
  \frac{z_m(s)}{\sum_{h=1}^M\int_0^Tz_h(s)D^h_s m(T)ds} &=& u_m(s) \\
  \sum_{h=1}^M\int_0^Tz_h(s)D^h_s m(T)ds &\neq& 0,\quad \mathrm{a.s.}
\end{eqnarray*}
\end{prop}
\begin{proof}
Compute
\begin{equation}
D_s^ha(T)=a'(T)D^h_s m(T)\quad h=1,\dots,M.
\end{equation}
Multiply the above equation by $z_h(t)$ -- so that $\mathbf{z}\in
\mathrm{dom}(\delta^{\mathrm{Sk}})$ -- and by $G_k$; then sum for
all $h=1,\dots,M$ and integrate:
\begin{equation}
\sum_{h=1}^M\int_0^TG_kz_h(s)D_s^ha(T)ds=\sum_{h=1}^M\int_0^TG_kz_h(s)a'(T)D^h_s
m(T)ds.
\end{equation}
Due to the definition of $\mathbf{u}$ and to the fact that
$a'(T)G_k$ does not depend on $s$ we can write
%\begin{equation}\label{eq:proofMall}
%a'(T)G_k=\sum_{m=1}^M\int_0^T\frac{z_m(s)D_s^mm(T))ds}{\sum_{h=1}^M\int_0^Tz_h(s)D^h_s
%m(T)ds},\quad k=1,\dots,M,
%\end{equation}
\begin{equation}\label{eq:proofMall}
a'(T)G_k=\sum_{m=1}^M\int_0^Tu_m(s)G_kD_s^ma(T))ds.\quad k=1,\dots,M
\end{equation}
%
%Define $u_m(s)$ as
%\begin{equation}
%u_m(s)=\frac{z_m(s)}{\sum_{h=1}^M\int_0^Tz_h(s)D^h_s m(T)ds},\quad
%m=1,\dots,M
%\end{equation}
Finally compute the expected value of both sides of
(\ref{eq:proofMall})
\begin{equation}
\mathbb{E}\left[a'(T)G_k\right]=\mathbb{E}\left[\sum_{m=1}^M\int_0^Tu_m(s)G_kD_s^ma(T)ds\right].
\end{equation}
so that by duality
\begin{equation}\label{DeltaMulti}
\Delta_k =
\mathbb{E}\left[a(T))\delta^{\mathrm{Sk}}(G_k\mathbf{u})\right]\quad
k=1,\dots,M.
\end{equation}
and this concludes the proof.
\end{proof}
Proposition \ref{prop:1} allows a certain flexibility in choosing
either the process $\mathbf{u}$, or better $\mathbf{z}$. We consider
$z_h=\alpha_k\delta^{\mathrm{Kr}}_{hk}$; $h,k=1,\dots,M$,
$\alpha_k=1, \forall k$. Namely, in order to compute the $k$-th
delta we consider only the $k$-th term of the Skorohod integral
reducing the computational cost. In particular, this choice is
motivated by the fact that in this way the localization technique
needs to control only $\delta^{\mathrm{Sk}}_{k}(\cdot)$ and then
only the $k$-th component of $\mathbf{W}(t)$. Then we define $L_k$
and calculate for $k=1,\dots,M$
\begin{eqnarray}\label{prova}
L_k\!\!\!&=&\!\!\!\int_0^TD_s^km(T)ds=\sum_{i=1}^M\sum_{j=1}^Nw_{ij}S_i(t_j)\int_0^{t_j}\sigma_{ik}(s)ds,\label{prova1}\\
\int_0^TD^k_sG_kds\!\!\! &=&\!\!\!
\sum_{j=1}^Nw_{jk}S_k(t_j)\int_0^{t_j}\sigma_{kk}(s)ds=\sum_{j=1}^Nw_{jk}S_k(t_j)\int_0^{t_j}\sigma_{k}(s)ds,\label{prova2}\\
\int_0^TD^k_sL_kds\!\!\!
&=&\!\!\!\sum_{j=1}^Nw_{ij}S_i(t_j)\left(\int_0^{t_j}\sigma_{ik}ds\right)^2,\label{prova3}
\end{eqnarray}
and hence
\begin{equation}\label{DelAsMultiDisCorr}
\Delta_k =
\mathbb{E}\left[a(T)\delta_k^{\mathrm{Sk}}\left(\frac{G_k}{L_K}\right)\right],\quad
k=1,\dots,M.
\end{equation}
Due to the properties of the Skorohod integral we have for
$\,k=1,\dots,M$
\begin{equation}
\delta_k\left(\frac{G_k}{L_K}\right)=\frac{G_k}{L_K}\,W_k(T)-
\frac{1}{L_k^2}\left(L_k\int_0^TD^k_sG_kds-G_k\int_0^TD^k_sL_kds\right),
\end{equation}
%\begin{eqnarray}
%\int_0^TD^k_sAds
%=\frac{\sum_{j=1}^Nw_{jk}S_k(t_j)\int_0^{t_j}\sigma_{kk}(s)ds\sum_{i=1}^M\sum_{j=1}^Nw_{ij}S_i(t_j)\int_0^{t_j}\sigma_{ik}(s)ds}
%{x_k\left(\sum_{i=1}^M\sum_{j=1}^Nw_{ij}S_i(t_j)\int_0^{t_j}\sigma_{ik}(s)ds\right)^2}+\nonumber\\
%-\frac{\sum_{j=1}^Nw_{jk}S_k(t_j)\sum_{i=1}^M\sum_{j=1}^Nw_{ij}S_i(t_j)\left(\int_0^{t_j}\sigma_{ik}ds\right)^2}
%{x_k\left(\sum_{i=1}^M\sum_{j=1}^Nw_{ij}S_i(t_j)\int_0^{t_j}\sigma_{ik}(s)ds\right)^2}.\nonumber\\
%\end{eqnarray}
With another choice of $\mathbf{z}$, for instance $z_h=\alpha_h$,
$\Delta_k$ would depend linearly on the whole $M$-dimensional BM,
making the localization technique less efficient.

We investigate the applicability of the RQMC approach to estimate
the expected value in Equation (\ref{DelAsMultiDisCorr}) for
$k=1,\dots,M$. We assume the same input parameters as in Section
\ref{sec:Num} and generate the trajectories (the values $S_i(t_j),
i=1,\dots,M, j=1,\dots,N$) in Equations (\ref{prova}-\ref{prova3})
as in Section \ref{sec:Num}. Moreover, we consider $\alpha_{im}$ as
the elements of the CH matrix associated to $\rho_{im}$,
$i,m=1,\dots,M$.
\begin{table}[tbp] \centering
\caption{At-the-money estimated $\Delta$'s ($10^{-2}$) and errors
($10^{-4}$) with RQMC.}\label{DeltaErr} \vspace{10pt}
\begin{tabular}{c}
         $\rho = 0\%$ \\
    \begin{tabular}{cccccccc}
                \multicolumn{2}{c}{LT} & \multicolumn{2}{c}{KPA}& \multicolumn{2}{c}{PCA}& \multicolumn{2}{c}{CH}\\
                \hline
                $\Delta$ &RMSE &$\Delta$ & RMSE & $\Delta$ &RMSE &$\Delta$ & RMSE\\
                \hline
                6.1832&0.80&6.1820&1.10&6.1808&0.86&6.2060&1.50\\
                6.2024&0.75&6.2126&0.90&6.2016&0.86&6.2250&1.10\\
                6.2305&0.85&6.2340&0.87&6.2341&0.99&6.2530&1.20\\
                6.2667&0.75&6.2701&0.82&6.2699&0.91&6.2830&1.40\\
                6.3081&0.60&6.3133&0.92&6.3093&0.96&6.3270&1.20\\
                6.3569&0.55&6.3595&0.97&6.3598&0.83&6.3750&1.10\\
                6.4107&0.50&6.4141&0.93&6.4103&0.78&6.4329&1.20\\
                6.4709&0.50&6.4744&0.91&6.4677&0.84&6.4920&1.40\\
                6.5338&0.50&6.5390&0.93&6.5325&0.93&6.5530&1.30\\
                6.6001&0.65&6.6060&0.78&6.6000&0.96&6.6120&1.10\\
                 \hline
            \end{tabular}
\\\\
    \end{tabular}
    \begin{tabular}{c}
        $\rho = 40\%$\\
        \begin{tabular}{cccccccccccc}
            \multicolumn{2}{c}{LT} & \multicolumn{2}{c}{KPA}& \multicolumn{2}{c}{PCA}& \multicolumn{2}{c}{CH}\\
            \hline
            $\Delta$ &RMSE &$\Delta$ & RMSE & $\Delta$ &RMSE &$\Delta$ & RMSE\\
            \hline
            5.47830&0.056&5.48410&0.110&5.48050&0.130&5.46767&1.100\\
            5.53510&0.062&5.54050&0.110&5.53740&0.140&5.52457&1.200\\
            5.59430&0.054&5.60020&0.110&5.59660&0.130&5.58730&1.200\\
            5.65440&0.062&5.66120&0.120&5.65680&0.130&5.64031&1.000\\
            5.71680&0.075&5.72330&0.130&5.71790&0.140&5.70994&1.100\\
            5.78130&0.082&5.78850&0.120&5.78410&0.120&5.77038&1.300\\
            5.84840&0.077&5.85320&0.096&5.85060&0.120&5.83233&1.200\\
            5.91560&0.082&5.92110&0.110&5.91790&0.130&5.90019&1.100\\
            5.98490&0.052&5.99070&0.093&5.98680&0.120&5.97059&1.000\\
            6.05470&0.059&6.06050&0.110&6.05680&0.130&6.04613&1.200\\

            \hline
        \end{tabular}
    \end{tabular}
\end{table}
Table \ref{DeltaErr} compares the deltas obtained with RQMC only
with the same number of scenarios as in Section \ref{sec:Num}. We
apply the same LT construction used to estimate the price of the
option and not the one for the integrand function in Equation
(\ref{DelAsMultiDisCorr}). This would not seem to be the optimal
choice, but if we would have applied the LT for the integrand
function in Equation (\ref{DelAsMultiDisCorr}) $M=10$ decomposition
matrices (one for each delta) would be required. This setting would
have increased the CPU time to obtain the LT to at least $1/3$ (even
higher due to the larger number of terms to compute) of the total
time making the estimation less convenient. Table \ref{DeltaErr}
shows that the PCA, LT and KPA approaches perform almost equally in
terms of RMSEs, with the LT giving only slightly better results in
the uncorrelated case. In terms of computational cost the KPA
performs better than the PCA. The CH construction displays RMSEs
that are even $10$ times higher. As explained before, $g$ can be
considered a good approximation for the payoff function in Equation
(\ref{DelAsMultiDisCorr}) but in the Malliavin expression $a(T)$ is
multiplied by a random weight that depends on the Gaussian vector
$\mathbf{Z}$. In contrast, the PCA and the KPA concentrate most of
the variation in the first dimensions of $\mathbf{Z}$. This is the
explanation of the almost equal accuracy of the LT, the PCA and the
KPA.

%%%%%%%%%%%%%%%%%%%%%%%%%%%%%%%%%%%%5555
%%%%%    CONCLUSION
%%%%%%%%%%%%%%%%%%%%%%%%%%%%%%%%%%%%5555
\section{Conclusions}\label{sec:concl}
We have considered the problem of computing the fair price and the
deltas of high-dimensional Asian basket options in a BS market with
time-dependent volatilities. In order to extend the QMC superiority
to high dimensions it is necessary to employ path-generation
techniques with the main purpose to reduce the nominal dimension.
The LT and the PCA constructions try to accomplish this task by the
concept of ANOVA. In the case of time-dependent volatilities in the
BS economy the computational cost of the LT and the PCA cannot be
reduced relying on the properties of the Kronecker product and the
computation is more complex. We have presented a new and fast CH
algorithm for block matrices that remarkably reduces the
computational burden making the LT construction even more convenient
than the PCA. We have introduced a new path-generation technique,
named KPA, that in the applied setting, is as accurate as the PCA
and is even more convenient with respect to the computational cost.
In addition, we proved that the KPA enhances RQMC for the estimation
of the fair price and the calculation of the deltas of Asian basket
options in a BS model with time-dependent volatilities. In this
setting the KPA provides the same accuracy of the LT in the case of
correlated asset returns and in the estimation of the deltas.
Finally, concerning the computation of the sensitivities, we have
extended the procedure adopted by Montero and Kohatsu-Higa
\cite{MKH2003}, based on the Malliavin Calculus,  to the
multi-assets case. All these results can be easily applied to
stochastic and local volatility models that are based on the mixture
of multi-dimensional dynamics for basket options, as done in Brigo
et al.~\cite{BMR}.
%%%%%%%%%%%%%%%%%%%%%%%%%%%%%%%%%%%%%%%%%%%%%%%%%%%%%%%%%%%%%%
%%%%%%%%%%%%%%%%%%%%%%%%%%%%%%%%%%%%%%

%%%%%%%%%%%%%%%%%%%%%%%%%%%%%%%%%%%%%%
%%%%%%%%%%%%%%%%%%%%%%%%%%%%%%%%%%%%%%%%%%%%%%%%%%%%%%%%%%%%%%%%%%%%%%
% Choose one of these bibliography styles, comment out the other:
    % Numbered references, citations like [2]
\bibliographystyle{unsrt}
\bibliography{pierbib}    % Small subset of database
\end{document}